\documentclass[12pt]{article}  
\usepackage{fullpage}     
\usepackage{epsfig}
\usepackage{subfigure}
\usepackage{latexsym}
\newtheorem{theorem}{Theorem}[section]
\newtheorem{corollary}{Corollary}[section]
\newtheorem{lemma}{Lemma}[section]
\newtheorem{proposition}{Proposition}[section]
\newenvironment{proof}{\par\noindent{\it Proof:\/}}{\hfill $\Box$}
\newenvironment{remark}{\par\noindent{\bf Remark:\/}}{\hfill $\Box$}

\newcommand{\Rank}{\mbox{\sf rank}} 
\newcommand{\Select}{\mbox{\sf select}} 
\newcommand{\fullrank}{\mbox{\sf fullrank}}
\newcommand{\Rankm}{\mbox{\sf rankm}} 
\newcommand{\Rankmplus}{\mbox{\sf fullrankm}} 
\newcommand{\Selectm}{\mbox{\sf selectm}} 
\newcommand{\Sum}{\mbox{\sf sum}} 
\newcommand{\Pred}{\mbox{\sf pred}} 

\def\floor#1{\left\lfloor #1\right\rfloor}
\def\ceil#1{\left\lceil #1\right\rceil}

\def\comment#1{}

\begin{document}

\title{\Large {Succinct Indexable Dictionaries with
Applications to Encoding $k$-ary Trees, Prefix Sums and Multisets
\thanks{Research supported by UK-India Science and Technology 
Research Fund project number 2001.04/IT and EPSRC grant GR L/92150.}}}
\author{Rajeev Raman\thanks{Department of Computer Science,
University of Leicester, Leicester LE1 7RH, UK. \texttt{r.raman@mcs.le.ac.uk}}
\and
Venkatesh Raman\thanks{Institute of Mathematical Sciences, Chennai, 
India 600 113. \texttt{vraman@imsc.res.in}}
\and
S. Srinivasa Rao\thanks{School of Computer Science, University of Waterloo, 
Waterloo N2L 3G1, Canada. \texttt{ssrao@uwaterloo.ca}} }
\date{}

\maketitle

\begin{abstract}
We consider the {\it indexable dictionary\/} problem, which consists 
of storing a set $S \subseteq \{0,\ldots,m-1\}$ for some integer $m$,
while supporting the operations of $\Rank(x)$, which returns the
number of elements in $S$ that are less than $x$ if $x \in S$, and
$-1$ otherwise; and $\Select(i)$ which returns the $i$-th smallest
element in $S$. We give a data structure that supports both operations in $O(1)$ time on
the RAM model and requires ${\cal B}(n,m) + o(n) + O(\lg \lg m)$ bits
to store a set of size $n$, where ${\cal B}(n,m) = \ceil{\lg {m
\choose n}}$ is the minimum number of bits required to store any
$n$-element subset from a universe of size $m$.  Previous dictionaries
taking this space only supported (yes/no) membership queries in $O(1)$
time. In the cell probe model we can remove the $O(\lg \lg m)$
additive term in the space bound, answering a question raised by Fich
and Miltersen, and Pagh.  

We present extensions and applications of our indexable dictionary 
data structure, including:
\begin{itemize}
\item an information-theoretically optimal representation of a $k$-ary 
cardinal tree that supports standard operations in constant time,
\item a representation of a multiset of size $n$ from $\{0,\ldots,m-1\}$
in ${\cal B}(n,m+n) + o(n)$ bits that supports
(appropriate generalizations of) $\Rank$ and $\Select$ operations in
constant time, and
\item a representation of a sequence of $n$ non-negative integers 
summing up to $m$ in ${\cal B}(n,m+n) + o(n)$ bits that 
supports prefix sum queries in constant time.
\end{itemize}

\end{abstract}

\section{Introduction~}\label{intro}

Given a set $S$ of $n$ distinct keys from the universe
$\{0,\dots,m-1\}$, possibly the most fundamental data structuring
problem that can be defined for $S$ is the {\it dictionary\/} problem:
to store $S$ so that {\it membership\/} queries of the form ``Is $x$
in $S$?'' can be answered quickly.  In his influential paper
\cite{Yao}, Yao considered the complexity of this problem and showed
that the sorted array representation of $S$ is the best possible for
this problem, if one considers a suitably restricted class of
representations. Since membership queries take $\Omega(\lg n)$ time to
answer using a sorted array\footnote{$\lg x$ denotes $\log_2 x$.}, 
a number of researchers have developed
representations based on hashing that answer membership queries in
constant time (see e.g. \cite{Yao,Tarjan-Yao,FKS,BM,Pagh}).

However, one extremely useful feature present in the sorted array
representation of $S$ is that, given an index $i$, the $i$-th smallest
element in $S$ can be retrieved in constant time.  Also, when the
presence of an element $x$ has been established in a sorted array, 
we know the {\it rank} of $x$, i.e., the number of elements 
in $S$ that are less than $x$. 
Schemes based on hashing work by ``randomly scattering'' keys, and do
not intrinsically support such operations.  It is natural to ask
whether one can represent $S$ in a way that combines the speed of hash
tables with the additional functionality of sorted arrays.  We
therefore consider the problem of representing $S$ to support the
following operations in constant time:
\begin{description}
\item[$\Rank(x, S)$] Given $x \in \{0,\ldots,m-1\}$, 
return $-1$ if $x \not \in S$ and
 $|\{y \in S | y < x \}|$ otherwise, and

\item[$\Select(i, S)$] Given $i \in \{1,\dots,n\}$, 
return the $i$-th smallest element in $S$.
\end{description}
When there is no confusion, we will omit the set $S$ from the
description of these operations.  We call this the {\it indexable
dictionary\/} problem, and a representation for $S$ where both these
operations can be supported in constant time an {\it indexable
dictionary representation}. 



Our interest lies in {\it succinct\/} representations
of $S$, whose space usage is close to the information-theoretic lower
bound.  Motivated by applications to very large data sets, as well as
by applications to low-resource systems such as handheld and embedded
computers, smart cards etc., there has been a renewal of interest in
succinct representations of data
\cite{BM,BDMRRR,jacobsonthesis,GV00,HT01,MR,MRR,Pagh,SS}.  In the
context of this paper, the information-theoretic lower bound is
obtained by noting that as there are $m \choose n$ subsets of size $n$
from $\{0,\ldots,m-1\}$, one cannot represent an arbitrary set of $n$
keys from $\{0,\ldots,m-1\}$ in fewer than ${\cal B}(n,m) =
\ceil{\lg{m \choose n}}$ bits, and
we seek representations that use space close to  ${\cal B}(n,m)$.

As ${\cal B}(n,m) = n \lg(em/n) - O(\lg n) - \Theta(n^2/m)$ \cite{Pagh}, a sorted array
representation of $S$, which takes $n \ceil{\lg m}$ bits, can be
significantly larger than the information-theoretic lower bound.
Brodnik and Munro \cite{BM} were the first to give a succinct
representation that supported constant-time membership queries.  Pagh
\cite{Pagh} improved the space bound to ${\cal B}(n,m) + o(n) + O(\lg
\lg m)$ bits, while continuing to support membership queries in
constant time.  Raman and Rao \cite{RR} considered {\it dictionaries
with rank}, which support constant-time $\Rank$ queries and gave a
representation requiring $n \ceil {\lg m} + O(\lg \lg m)$ bits of
space; this is better than augmenting Pagh's data structure with $n
\ceil{\lg n}$ bits of explicit rank information.  Raman and Rao's data
structure can also support $\Select$ queries using $n( \ceil{\lg m} +
\ceil{\lg n}) + O(\lg\lg m)$ bits, but this is nearly $2 n \lg n$ bits
more than necessary.  All the above papers \cite{BM,Pagh,RR} assume
the standard {\it word RAM\/} model with word size $\Theta(\lg m)$
bits \cite{AHU,Hagerup:stacs}; unless specified otherwise this is our
default model.

\subsection{Our results} \label{sec:results}

\subsubsection{Indexable dictionaries}

We give an indexable dictionary representation that
requires ${\cal B}(n,m) + o(n) + O(\lg \lg m)$ bits to store a set of
size $n$ from $\{0,\ldots,m-1\}$.  
Modifying this data structure, we get an indexable dictionary representation
that requires ${\cal B}(n,m) +o(n)$ bits and supports operations in $O(1)$ time
in the {\it cell probe\/} model \cite{Yao} with
word size $\Theta(\lg m)$. The significance of this modest
improvement in space usage is as follows.  Since ${\cal B}(n,m) +o(n)
\le n \ceil{\lg m}$ for all $n$ larger than a sufficiently large
constant, this result shows that $n$ words of $\ceil{\lg m}$ bits
suffice to answer membership queries in constant time on a set of size
$n$, and answers a question raised by Fich and Miltersen \cite{FM}
and Pagh \cite{Pagh}.  By contrast, Yao showed that if the $n$ words
must contain a permutation of $S$, then membership queries cannot be
answered in constant time \cite{Yao}.  

\subsubsection{Applications of indexable dictionaries}
Using the indexable dictionaries, we obtain the following results:
\begin{itemize}
\item 
A \emph{$k$-ary cardinal tree} is a rooted tree, each node of which has $k$
positions labeled $0, \ldots, k-1$, which can contain edges to
children.  The space lower bound for representing a $k$-ary cardinal
tree with $n$ nodes is ${\cal C}(n,k) = \left \lceil \lg \left (
\frac{1}{kn+1}{{kn + 1} \choose n} \right ) \right \rceil$ \cite{GKP}.
Note that ${\cal C}(n,k) = (k \lg k - (k-1)\lg (k-1))n - O(\lg (kn))$,
which is close to $n (\lg k + \lg e)$, as $k$ grows.
%
%
%
Benoit et al.  \cite{BDMRRR} gave a cardinal tree data structure that takes 
$(\lceil \lg k \rceil + 2)n + o(n) + O(\lg \lg k) = {\cal C}(n,k) + \Omega(n)$ 
bits and answers queries asking for parent, $i$-th child, child with 
label $i$, degree and subtree size in constant time.  We 
obtain an encoding for
$k$-ary cardinal trees taking ${\cal C}(n, k) + o(n) + O(\lg \lg k)$
bits, in which all the above operations, except the subtree size at a
node, can be supported in constant time.  
Both the above results on cardinal trees use the word RAM model with a
word size of $\Theta (\lg (k+n))$ bits.

\item Let $M$ be a multiset of $n$ numbers from $\{0,\ldots,m-1\}$.
We consider the problem of representing $M$ to support the following 
operations:

\begin{description}
\item[$\Rankm(x,M)$] Given $x \in U$, return $-1$ if $x \not \in M$ and
  $|\{y \in M | y < x \}|$ otherwise, and
  
\item[$\Selectm(i,M)$] Given $i \in \{1,\dots,n\}$, return the largest
  element $x \in M$ such that $\Rankm(x) \le i - 1$.
\end{description}

$\Rankm$ and $\Selectm$ are natural generalisations of $\Rank$
and $\Select$ to multisets.  It is easy to see that ${\cal B}(n, m+n)$
is a lower bound on the number of bits needed to represent such a
multiset, as there is a $1-1$ mapping between such multisets and sets
of $n$ elements from $\{0,\ldots,m+n-1\}$ \cite{elias}.  However, if
we transform a multiset into a set by this mapping, then $\Rankm$ and
$\Selectm$ do not appear to translate into $\Rank$ and $\Select$
operations on the transformed set.  Using some additional ideas, we
obtain a multiset representation that takes ${\cal B}(n, m+n) + o(n)
+ O(\lg\lg m)$ bits, and supports $\Rankm$ and $\Selectm$ in constant
time.  This result assumes a word size of $\Theta(\lg (m+n))$ bits.

Elias \cite{elias} previously considered the problem of 
representing multisets succinctly
while supporting  $\Selectm$ and the following
generalization of $\Rankm$:

\begin{description}
  
\item[$\Rankmplus(x)$] Given $x \in U$, return $|\{y \in M | y < x
  \}|$.
\end{description}

He considered the {\it bit-probe\/} model rather than the word RAM model,
and was concerned with average-case behaviour over all possible
operations.  Our results on FIDs (see below) have consequences
for this version of the problem.
\end{itemize}

%

\subsubsection{Fully indexable dictionaries and prefix sums}

We also give a subroutine that appears to be of independent interest.
Given a sequence $\sigma$ of $m$ bits, define the following 
operations, for $b \in \{0,1\}$: 
(a) $\Rank_b(i)$ -- count the number of $b$'s before the position $i$ in 
$\sigma$, and
(b) $\Select_b(i)$ -- find the position of the $i$-th $b$ in $\sigma$.
It is shown in
\cite{clarkthesis,MRR} how to represent $\sigma$ in $m + o(m)$ bits
and support these four queries in constant time.  This data structure
is a fundamental building block in a large number of succinct data
structures \cite{BDMRRR,Tarjan-Yao,Pagh,GV00,MRR,munro}.

One can also view $\sigma$ as the 
characteristic vector of a subset $S$ of $n$ keys from
$U = \{0,\ldots,m-1\}$, and define a {\it fully indexable dictionary
(FID)\/} representation of $S$ to be one that supports the operations
$\Rank(x,S)$, $\Select(i,S)$, $\Rank(x,\bar{S})$ and
$\Select(i,\bar{S})$ all in constant time, where $\bar{S} = U
\setminus S$ is the complement of the set $S$.
It is easy to see that an FID representation is functionally
equivalent to a bit-vector supporting $\Rank_{0/1}$ and $\Select_{0/1}$.
Extending a result due to Pagh \cite{Pagh}, we give an FID
representation for $S$ that takes ${\cal B}(n,m) + O((m \lg \lg m)/\lg
m)$ bits.  This is always at most $m + o(m)$ bits, but it may be
substantially less: for example, whenever $m / \sqrt{\lg m} \le n \le
m (1-1/\sqrt{\lg m})$, the space usage is at most ${\cal B}(n,m) +
o(n) = (1+o(1)){\cal B}(n,m)$ bits.
We give the following application of this result:

\begin{itemize}
\item We can store a multiset $M$ of $n$ values from $\{0,\ldots,m-1\}$
to support $\Selectm$ (but not $\Rankm$) in constant time using 
${\cal B}(n,m+n) + o(n)$ bits. 
%
%
Another way of stating this result is that we can represent a sequence 
of $n$ non-negative integers $X = x_1,\ldots,x_n$, such that 
$\sum_{j=1}^n x_j = m$, so that the query $\Sum(i,X)$, which returns 
$\sum_{j=1}^{i} x_j$, can be answered in constant time using
${\cal B}(n,m+n) + o(n)$ bits.

\end{itemize}

The problem of representing integers compactly so that their prefix
sums can be computed efficiently has been studied by a number of
researchers including \cite{elias,GV00,HT01,Pagh,Tarjan-Yao}.  Our
solution is more space-efficient than all of these.  The result of
Grossi and Vitter \cite[Lemma 2]{GV00}, which is based on Elias's
ideas, is the previously most space-efficient one and requires
$n(\ceil{\lg m} - \floor{\lg n} + 2) + o(n)$ bits to represent $n$
non-negative integers adding up to $m$, where $m\ge n$.  In most
cases, this will be $\Theta(n)$ bits more than optimal.  When $n$ and
$m$ are not powers of 2, the ceilings and floors are a source of
non-optimality; for example, take $m=n$ with $m$ not a power of 2;
Grossi and Vitter's method requires $3n + o(n)$ bits in the worst
case, as opposed to the lower bound of ${\cal B}(n,2n) = 2n - O(\lg
n)$ bits.  Another source of non-optimality is that the constant 2 is
not optimal; for example, take $m = cn$ where $n$ and $m$ are powers
of 2 and $c > 1$.  Grossi and Vitter's method requires $(2 + \lg c) n
+ o(n)$ bits in the worst case, which can be easily shown to be at
least $(2 - (1+c)\lg((1+c)/c)) n = \Omega(n)$ bits more than optimal
(the difference tends to $(2 - \lg e)n$ as $c$ increases).  On the
other hand, our representation is always within $o(n)$ bits of
optimal.

%

\subsubsection{Lower bounds}
It is important to note that, appearances notwithstanding, some of the
space bounds above may actually be much larger than the
information-theoretic lower bound of ${\cal B}(n,m)$.  For
example, consider the space bound of ${\cal B}(n,m) + o(n) + O(\lg \lg
m)$ bits for storing a set $S$ of size $n$ from $\{0,\ldots,m-1\}$ in
an indexable dictionary representation.  If $n \le m/2$, ${\cal B} =
{\cal B}(n,m) \ge \max\{n,\lg m\}$ and this space bound is indeed
${\cal B}$ plus lower-order terms.  However, as $n$ gets very close to
$m$, ${\cal B}$ can be much smaller than the $o(n)$ term.  If we only
want to answer membership queries, we can assume $n \le m/2$ without
loss of generality: if $S$ has more than $m/2$ elements then we store
its complement and invert the answers.  However, in the indexable
dictionary problem, it is not clear how answering $\Rank$ and
$\Select$ queries on a set could help us to answer these queries on
its complement in constant time.  In fact, we note that if we could
store a set $S$ in ${\cal B}^{O(1)}$ bits 
for all $n$ and $m$, and support $\Select$ (or $\Rank$) in constant time, then we
could also support $\fullrank$ queries on $S$ in constant time using
${\cal B}^{O(1)}$ bits. Here $\fullrank(x,S)$ returns the rank of $x$
in $S$ for any $x \in U$. It is known that in general, $\fullrank$
queries cannot be answered in constant time in the RAM model (or even
in the cell probe model) while using $n^{O(1)}$ words of $(\lg
m)^{O(1)}$ bits each \cite[Corollary 3.10]{BeameFich}.  
Thus, many of our space bounds
are of necessity not information-theoretically optimal 
in some cases; one exception is the space
bound for $k$-ary trees, which is optimal for all $k \ge 2$.

\subsection{Techniques used}
The main ingredient in our indexable dictionary representation is {\it
most-significant-bit first (MSB) bucketing}. The idea is to apply a
trivial top-level hash function to the keys in $S$, which simply takes
the value of the $t$ most significant bits of a key.  As we can omit
the $t$ most significant bits of all keys that ``hash'' to the same
bucket, space savings is possible.  A similar idea was used by Brodnik
and Munro \cite{BM} in their succinct representation of sets.  A major
difference between our approach and theirs is that they store explicit
pointers to refer to the representation of buckets, which uses more
space than necessary (and hence constrains the number of
buckets). Instead, we use a succinct representation of the prefix sums
of bucket sizes that not only provides the extra functionality needed
for supporting $\Rank$ and $\Select$, but also uses significantly less
space.  The related technique of {\it quotienting\/}
\cite{Pagh,Cleary} stores only the ``quotients'' of keys that are
mapped to a bucket by a standard hash function (e.g. those of
\cite{FKS}).  The crucial difference is that MSB bucketing preserves
enough information about the ordering of keys to allow us to maintain most
of the rank information using negligible extra space.

Other ideas relevant to the indexable dictionary representation are
range reduction (\cite{FKS} and others), distinguishing bits
(\cite{AFK} and others) and techniques for compactly representing hash
functions for several subsets of a common universe developed in
\cite{BDMRRR}.  Our $k$-ary tree representation 
does not encode the tree structure explicitly, a feature shared with
the representation of \cite{Bonsai}.

\subsection{Organization of the paper}
The remainder of this paper is organized as follows.  In Section~2, we
give some building blocks that will later be used in our main results.
Extending the dictionary with rank of \cite{RR,BDMRRR}, we first give
a simple indexable dictionary that uses about $2 n \lg n$ bits more
than necessary. Then we show the connection between fully indexable
dictionaries and prefix sum data structures and give some simple
representations for both. These are then used in Section 3, coupled
with MSB bucketing, to obtain an improved result on indexable
dictionaries, which reduces the space wastage to about $O(n)$ bits.

In Section 4, we first develop a ${\cal B}(n, m) + O((m \lg \lg m)/\lg
m)$-bit fully indexable dictionary representation,
extending a result of
Pagh \cite{Pagh}.  Using this and our result from Section 3, we
obtain our main result: an indexable dictionary taking ${\cal B}(n, m)
+ o(n) + O(\lg \lg m)$ bits.  In Section 5, we remove the $O(\lg \lg
m)$ term in the space bound by moving to the cell probe model, giving
a representation that takes ${\cal B}(n,m) +o(n)$ bits.
Section 6 gives some applications of our succinct dictionaries to
representations of multiple dictionaries, $k$-ary trees, multisets and
prefix sums. Section 7 makes some observations about the difficulty of
achieving optimal space for all values of the input parameters.
Section 8 recapitulates the main results and gives some open problems.

\section{Preliminaries}

In this section, we first establish connections between FIDs and prefix 
sums, and  
we end with simple representations of multiple indexable dictionaries 
and prefix sums.

In what follows, if $f$ is a function defined from a finite set $X$ to
a finite totally ordered set $Y$, by $|| f ||$, we mean $ \max \{ f(x)
: x \in X \}$. We use the notation $[m]$ to denote the set $\{0, 1,
\ldots, m-1 \}$.

\subsection{Fully Indexable Dictionaries and Searchable Prefix Sums}
Given a set $S \subseteq U$, recall that a {\it fully indexable
dictionary (FID)\/} representation for $S$ supports $\Rank$ and
$\Select$ operations on both $S$ and its complement $\bar{S} = U
\setminus S$ in $O(1)$ time.  FIDs are essential to our data structure
as they are intimately related to operations on prefix sums, as we
note below.

Given a sequence $X$ of $n$ non-negative integers $x_1, \ldots, x_n$
such that $\sum_{i=1}^n x_i =m$, the {\it searchable prefix sum
problem\/} is to find a representation of this sequence that supports the
following operations in constant time:

\begin{description}
\item[$\Sum(i, X)$] Given $i \in \{1,\dots,n\}$, 
return $\sum_{j=1}^i x_j$

\item[$\Pred(x, X)$] Given $x \in [m]$, return 
 $\max \{i \leq n | \sum_{j=1}^i x_j < x \}$.
\end{description}

We call a data structure that stores the sequence $X$ to support the
queries in constant time an $(n,m)$-searchable prefix sum data structure.
We now make the connection between FIDs and the searchable prefix sums
problem \cite{elias}.
%
%
\begin{lemma}
\label{prefromset}
Suppose that there is an FID representation for a given set $S \subseteq
U$ that uses $f(|S|,|U|)$ bits.  Then given a sequence $X$ of non-negative
integers $x_1, x_2, \ldots, x_n$, such that $\sum_{i=1}^n x_i =m$ there
is an $(n,m)$-searchable prefix sum data structure for $X$ 
using $f(n,m+n)$ bits.
\end{lemma}
\begin{proof}
Consider the following $m+n$ bit representation of the sequence $X$.
For $i=1$ to $n$, represent $x_i$ by $x_i$ $0$s followed
by a $1$.  Clearly this representation takes $m+n$ bits since it has
$m$ $0$s and $n$ $1$s.  View this bit sequence as the characteristic
vector of a set $S$ of $n$ elements from the universe $[m+n]$.
Represent $S$ as an FID using $f(n,m+n)$ bits. It is easy to verify
that $\Pred(x,X) = \Select(x, \bar{S}) - x +1$ and $\Sum(i,X) =
\Select(i,S) - i + 1$.  (Recall that $[m+n]$ begins with $0$.)
\end{proof}
\begin{lemma}
\label{ranksel}
Given a set $S \subseteq [m]$, there is an FID for $S$ taking
$m + o(m)$ bits.
\end{lemma}
\begin{proof}
Consider the characteristic vector of $S$, which is a bit-vector of 
length $m$. It is shown in \cite{jacobsonthesis,clarkthesis,munro,MRR} 
how to represent 
this bit-vector using $m + o(m)$ bits, to support the queries
$\Rank_b(i)$ and $\Select_b(i)$ in constant time, for $b \in \{0,1\}$.
It is easy to verify that these operations on the characteristic 
vector suffice to support FID operations on $S$: for example, $\Rank_1(j)$
is given by $\Rank(j,S)$ if the $j$-th bit is a 1, and by $j -
\Rank(j,\bar{S}) - 1$ otherwise. 
\end{proof}

The following lemma is an immediate consequence of
Lemma~\ref{prefromset} and Lemma~\ref{ranksel}:
\begin{lemma}
\label{lem:multis}
A sequence $S$ of $n$ non-negative numbers whose total sum is $m$ can
be represented using $m + n + o(m+n)$ bits to support $\Sum$ and $\Pred$ 
operations in constant time.
\end{lemma}


\subsection{A Simple Indexable Dictionary}

We now give a simple indexable dictionary representation for a set $S
\subseteq [m^*]$, based on perfect hashing schemes for membership
\cite{FKS,SS,Pagh}.  These perfect hashing schemes begin with
finding a {\it universe reduction\/} function $f : [m^*] \rightarrow
[|S|^2]$ such that $f$ is $1-1$ on $S$.  A problem with such an
approach is that $f$ requires $\Omega(\lg \lg m^{*})$ bits to
represent, which can be a significant overhead for sets which are very
small but nonetheless not constant-sized.  This 
overhead becomes significant
if we need to store several sets in the data structure, as we pay the
overhead repeatedly for each set.  To reduce this overhead we use the
idea of \cite{BDMRRR}, which is to note that we do not need $f$ to
bring the universe size as far down as $|S|^2$ for small sets, thereby
allowing the same $f$ to be used for several (small) sets.

Thus, it makes sense to talk about representing the set $S$, but
excluding the space cost of representing a universe-reduction
function.  In the following lemma, which is a simplification and
extension of a scheme from \cite{BDMRRR}, we use this approach.  Here
$h_S$ is the universe-reduction function and $q_S$ is a ``quotient''
function, which gives the information thrown away during the universe
reduction and is used to recover $x$ given $h_S(x)$.
%
%
%
\begin{lemma}\label{simple}
Let $m^*, n^* \ge 1$ be two given integers, and let $S \subseteq
[m^*]$ be a set of size at most $n^*$.  Suppose that we have access to two
functions $h_{S}$ and $q_{S}$, defined on $[m^*]$, satisfying the
following conditions:
\begin{enumerate}
\item $h_{S}$ is 1-1 on $S$.
\item $h_S$ and $q_S$ can be evaluated in $O(1)$ time, and
from $h_{S}(x)$ and $q_{S}(x)$ one can uniquely reconstruct $x$
in $O(1)$ time.
\item $||h_S|| = O((n^*)^2)$ if $|S| > \sqrt {\lg n^*}$ and 
$||h_S|| = O({(\lg n^*)}^c)$ for some constant $c> 0$ otherwise.
\item $\lceil \lg ||h_S|| \rceil + \lceil \lg ||q_S|| \rceil$
$= \lg m^* + O(1)$.
\end{enumerate}
Then we can represent $S$ using $|S| (\lg m^{*} + \lg |S| + O(1))$
bits and support $\Rank$ and $\Select$ in $O(1)$ time.  This assumes a
word size of at least $\lg \max\{m^*,n^*\}$ bits, and access to a
pre-computed table of $o(n^*)$ bits and a constant of $O(\lg n^*)$
bits that depends only on $||h_S||$, and that $m^*$ and $n^*$ are
known to the data structure.
\end{lemma}
\begin{proof}
Let $l = |S|$ and suppose that $S$ contains the elements 
$x_1 < x_2 < \ldots < x_{l}$.

If $l \le \sqrt{\lg n^*}$ then we write down
$h_S(x_1),\ldots,h_S(x_l)$ in fields of $b = \lceil \lg ||h_S||
\rceil$ bits each, followed by $q_S(x_1),\ldots,q_S(x_l)$ in fields of
$\lceil \lg ||q_S|| \rceil$ bits each.  This requires $|S| (\lg m^* +
O(1))$ bits.  To compute $\Rank(x)$ we calculate $h_S(x)$ and look for
a match in $h_S(x_1),\ldots,h_S(x_l)$.  This can be done in $O(1)$
time using standard techniques \cite{PS,FW,AHNR}, provided we have
available the integer constant $k$ that contains 1s in bit positions
$0, b, 2b, \ldots, b \cdot \lfloor (\lg n^*) /b \rfloor$, as well as
tables that enable us to compute, for every integer $x$ of $\lg n^*$
or fewer bits, the index of the most significant bit that is set to 1
(or, equivalently to compute $\floor{\lg x}$).
%
%
If we are unable to find an index $i$ such that $h_S(x) = h_S(x_i)$,
we return $-1$, otherwise we verify whether $q_S(x) = q_S(x_i)$. If
so, return $i-1$, otherwise return $-1$. To compute $\Select(i)$,
reconstruct $x_i$ from the values $h_S(x_i)$ and $q_S(x_i)$ and return
it.

If $l > \sqrt{\lg n^{*}}$, then let $S' = \{ h_S(x) | x \in S \}$.  
We create a {\it minimal\/} perfect hash function $f: [||h_S||]
\rightarrow [l]$ that is $1-1$ on $S'$.  As shown in \cite{SS,HT01},
there exists such a function $f$ that can be evaluated in $O(1)$ time
and that can be represented in $O(l + \lg \lg ||h_S||) = O(l + \lg \lg
n^*) = O(l)$ bits.  We also store two tables of size $l$. In the first
table $R$, for $1 \leq i \leq l$, we store the value $i$ in the
location $f(h_S(x_i))$ using a total of $l \ceil {\lg l}$ bits.  In
the second table $X$, we store the elements of $S$ in sorted order.
Now to answer $\Rank(x)$, we calculate $j = R[f(h_S(x))]$ and check if
$x = x_j$: if so, then $\Rank(x)$ is $j-1$, and is $-1$ otherwise.
Supporting $\Select$ is trivial since we have stored the $x_i$s in
sorted order in $X$. 
\end{proof}

The following lemma from \cite{BDMRRR} gives the space savings
obtained by combining universe reduction functions for different sets:
%
%
%
%
%
\begin{lemma}[\cite{BDMRRR}]
\label{lem:sharing}
Let $n^{*},m^{*}$ be as in Lemma~\ref{simple}, and let $0 \le i_1 <
i_2 < \ldots < i_s < n^{*}$ be a sequence of integers.  Let $S_{i_1},
S_{i_2}, \ldots, S_{i_s}$ be subsets of $[m^{*}]$ such that
$\sum_{j=1}^s |S_{i_j}| \le n^{*}$.  Then there exist functions
$h_{S_{i_j}}$ and $q_{S_{i_j}}$ for $j = 1,\ldots,s$ that satisfy the
conditions of Lemma~\ref{simple}, which can be represented in 
$o(n^{*}) + O(\lg\lg m^{*})$ bits in such a way that given $i_j$ we
can access $h_{S_{i_j}}$ and $q_{S_{i_j}}$ in constant time.
\end{lemma}
%
%
%
%

\section{Saving \mbox{$n \lg n$} bits using MSB bucketing}
\label{sec:optimal}

In this section, we first give a representation that takes
about $n \lceil \lg m \rceil$ bits to represent a set of size $n$ from
a universe of size $m$, and supports $\Rank$ and $\Select$ operations
in $O(1)$ time (Theorem~\ref{single}).  We then use this
representation to store multiple independent (but not necessarily
disjoint) dictionaries efficiently (Lemma~\ref{lem:multiple}).
%
%
\begin{theorem}\label{single}
There is an indexable dictionary for a set $S \subseteq [m]$, $|S|=n$,
that uses at most $n \ceil{\lg m} + o(n) + O(\lg \lg m)$ bits of space.
\end{theorem}
\begin{proof}
Our construction algorithm partitions $S$ using MSB bucketing,
recursing on large partitions. The base case of the recursion is
handled using Lemma~\ref{simple}. We get an overall space bound of $n
\ceil{\lg m}$ assuming the hypothesis of Lemma~\ref{simple} for each
application of this lemma. We then show how to support \Rank{} and
\Select{} in $O(1)$ time. Finally, we sketch how to use
Lemma~\ref{lem:sharing} to represent all functions used in
applications of Lemma~\ref{simple} using $o(n) + O(\lg \lg m)$ bits.

%
%

\begin{figure}
  \centering
  \epsfig{file=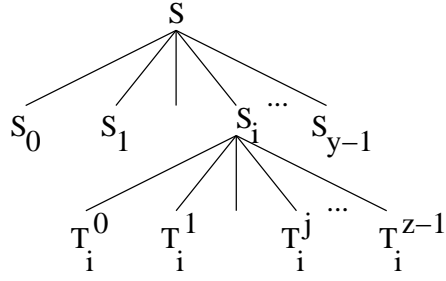}
  \caption{\label{fig:bucketing}
    Two level MSB bucketing}
\end{figure}

Let $t = \ceil{\lg m} - \ceil{\lg n}$, and let $c$ and $d$ be two
constants whose values are to be determined later.  If $n \le d$, then
we store the elements of $S$ explicitly, using $n \ceil{\lg m}$ bits,
and we are done.

Otherwise, if $n > d$, we partition the elements of $S$ according to
their top $\ceil{\lg n}$ bits.  This partitions $S$ into $y =
2^{\ceil{\lg n}} \le 2n$ sets denoted by $S_0, \dots, S_{y-1}$, where
$S_i$ consists of the last $t$ bits of all keys in $S$ whose most
significant $\ceil{\lg n}$ bits have value $i$, for $i \in [y]$.  We
store a representation of the sizes of these $y$ sets which takes $n +
y + o(n+y) \le 3n+o(n) < 4n$ bits (for sufficiently large $n$) 
using Lemma~\ref{lem:multis} and pad this
out to $4n$ bits. The representation of $S$ is obtained by
concatenating these $4n$ bits with the representations of each of the
$S_i$'s, for $i \in [y]$.

The representation of $S_i$, for $i \in [y]$, is obtained as follows.
Let $n_i = |S_i|$. If $n_i \le d$, we write down the elements of $S_i$
using $n_i t$ bits, and pad this out to $n_i (t+4+c)$ bits.
Otherwise, we again partition the elements of $S_i$ into $z =
2^{\ceil{\lg n_i}}$ sets according to their top $\ceil{\lg n_i}$ bits,
denoted as $T_i^0, \ldots, T_i^{z-1}$ (see Fig.~\ref{fig:bucketing}). 
We store a representation of
the sizes of these $z$ sets and pad this out to $4 n_i$ bits. Again,
the representation of $S_i$ is the concatenation of these $4 n_i$ bits
with the representations of each of the $T_i^j$s, for $j \in [z]$.

The representation of $T_i^j$, for $j \in [z]$, is obtained as
follows. If $|T_i^j| \le d$, then we write down its elements using
$|T_i^j| (t - \ceil{\lg n_i})$ bits, and pad this out to
$|T_i^j| (t + c)$ bits.  Otherwise, we store it using the
representation of Lemma~\ref{simple} (with $m^* = 2^{t - \ceil{\lg
n_i}}$ and $n^* = n$), padding this out to $|T_i^j| (t + c)$
bits if necessary. (Note that the representation of $T_i^j$ using
Lemma~\ref{simple} takes $|T_i^j| (t - \ceil{\lg n_i} + \lg |T_i^j| +
O(1))$ bits. Thus it is enough to choose $c$ to be equal to the
constant in the $O(1)$ term to guarantee that this is at most $|T_i^j|
(t+c)$ bits.)

When $S_i$ is partitioned, its representation takes $4 n_i +
\sum_{j=0}^{z-1} |T_i^j| (t+c) = 4 n_i + n_i (t+c) = n_i (t+4+c)$
bits. Thus the representation of $S_i$ takes $n_i (t+4+c)$ in either
case. Hence the length of the representation of $S$, when it is
partitioned, is $4n + \sum_{i=0}^{y-1} n_i (t+4+c) = 4n + n (t+4+c) =
n (t+8+c)$ bits. Thus, in either case, $S$ takes $n (t+8+c)$ bits.
This is at most $n \ceil{\lg m}$ bits, for sufficiently large $d$
(since $t = \ceil{\lg m} - \ceil{\lg n})$.

We now describe how the computation of $\Rank$ proceeds; $\Select$
works in a similar way. If $n \le d$, we apply the trivial algorithm
and return. Otherwise, we consider the first $4n$ bits of the
representation of $S$, which contains the representation of the sequence
$\sigma$ of the sizes of the buckets $S_i$, $i \in [y]$.  We extract
the the top $\ceil{\lg n}$ bits of the current key\footnote{Standard
techniques allow us to calculate $\ceil{\lg x}$ in constant time
\cite{FW}.}; suppose that these bits have value $i$.  Using
Lemma~\ref{lem:multis}, we calculate $\rho = \Sum(i-1, \sigma )$ and
$\rho' = \Sum(i, \sigma )$ in $O(1)$ time; note that $\rho' - \rho$ is
the size of the set $S_i$ to which the current key belongs.  The start
of the representation of $S_i$ is also easy to compute: it starts $4n
+ \rho (t + 4 + c)$ bits from the start of the representation of $S$.
We then remove the top $\ceil{\lg |T|}$ bits from the query key, add
the $\Rank$ of the resulting key in the set $S_i$ to $\rho$ and
return. Thus the problem reduces to finding the $\Rank$ of a key in
some set $S_i$.

If $|S_i| \le d$, then we apply the trivial algorithm to find the
$\Rank$ of a key in $S_i$. Otherwise, we apply a similar algorithm
as above to reduce the problem to finding the $\Rank$ of a key in some
set $T_i^j$. Again, if $|T_i^j| \le d$, then we apply the trivial
algorithm to find the $\Rank$. Otherwise, since $T_i^j$ is stored
using the representation of Lemma~\ref{simple}, we can support $\Rank$
in constant time. The overall computation is clearly constant-time.


It is easily verified that $n^* = n$ is an appropriate choice for all
applications of Lemma~\ref{simple} above.  We now verify that the
additional space required (in terms of the pre-computed table and
constants) is not excessive.  Firstly, the pre-computed table is of
size $o(n^*) = o(n)$ bits and is common to all applications of
Lemma~\ref{simple}.  At most $O(\lg n)$ constants are required, one
for each possible value of $b = \lceil \lg ||h_S|| \rceil$, which
require negligible space.

We now discuss the use of Lemma~\ref{lem:sharing} to represent the
functions for all the base-case sets.  The lemma requires that there
is a numbering of the sets using integers from $[n]$, but we can
simply take the number of a set to be the 
sums of the cardinalities of the sets whose indices
are less than its own index.  This information must be
computed anyway during $\Rank$ and $\Select$.  Finally, the space
required for representing the functions is $o(n) + O(\lg \lg m)$ bits.
This completes the proof of Theorem \ref{single}.  
\end{proof}
The following lemma is an easy extension of Theorem \ref{single}.
\begin{lemma}\label{lem:multiple}
Let $S_1, S_2 , \ldots, S_s$ all contained in $[m]$ be given sets with
$S_i$ containing $n_i$ elements, such that $\sum_{i=1}^s n_i =n$. Then
this collection of sets can be represented using $n \ceil {\lg m } +
o(n) + O(\lg \lg m)$ bits where the operations $\Rank (x, S_i)$ and
$\Select (j, S_i)$ can be supported in constant time for any $x \in
[m], 1 \leq j \leq n$ and $1\leq i \leq s$.  This requires that we
have access to a constant-time oracle which returns the prefix sums of
the $n_i$ values.
\end{lemma}
\begin{proof}
If we apply Theorem~\ref{single} directly to each set $S_i$ we get a
representation taking $\sum_{i=1}^s \left ( n_i \ceil{\lg m} + o(n_i)
+ O(\lg\lg m) \right ) = n \ceil{\lg m} + o(n) + O(s \lg \lg m)$ bits,
that supports $\Rank$ and $\Select$ on each set in $O(1)$ time. The
beginning of the representation of each set can be calculated using
the oracle supporting the prefix sum queries in constant time.  To get
the claimed space bound, we apply Theorem~\ref{single} to represent each
$S_i$, but with the modification that Lemma~\ref{lem:sharing} is used only
once across all applications of Theorem~\ref{single}.  
The only change this causes is that we need 
a global numbering (using indices bounded by $n$) of all base-case
sets created when applying Theorem~\ref{single} to the $S_i$'s.
Recall that when applying Theorem~\ref{single} to a particular set
$S_i$, we give each base-case set that is created a `local' number
bounded by $n_i$.  Thus, an appropriate global number for a base-level
set created when applying Theorem~\ref{single} to $S_i$ is just its
local number plus $\sum_{j=1}^{i-1} n_j$.  This gives the claimed
bound.  
\end{proof}

\section{Obtaining a sublinear lower-order term}

In this section, we develop the main result of the paper, namely, a
representation for an indexable dictionary taking ${\cal B}(n,m) +
o(n) + O(\lg \lg m)$ bits of space.  We begin by observing that the
bound of Theorem~\ref{single} is better than claimed: it is actually
${\cal B}(n,m) + O(n + \lg \lg m)$ bits.  The constant factor in the
$O(n)$ term can be improved by means of one more level of MSB
bucketing, as follows.  We place the keys into $2^{\floor{\lg n}}$
buckets based upon the first $\floor{\lg n}$ bits of each element.  We
represent the sizes of these buckets using at most $2n + o(n)$ bits
via Lemma \ref{lem:multis}. This partitions the given set into
multiple (up to $n$) sets which contain keys of $\ceil{\lg m} -
\floor{\lg n}$ bits each; the collection of sets is then represented
using the data structure of Lemma~\ref{lem:multiple}. The resulting
dictionary takes at most $n (\ceil{\lg m} - \floor{\lg n } + 2) + o(n)
+ O(\lg \lg m)$ bits and supports $\Rank$ and $\Select$ in constant
time.

Recalling the discussion on representing prefix sums in the
introduction, this bound is also non-optimal by $\Theta(n)$ bits in
many cases.  In addition to redundancy caused when $m$ and $n$ are not
powers of 2, the constant $2$ is not optimal.  For example, when $m =
cn$ for some constant $c>2$, the disparity in this case is $(2 -
c\lg(c/(c-1))) n$ bits, which tends again to about $(2 - \lg e) n$
bits for large $c$.  To bring the linear term of the space bound
closer to optimal, we place the keys into $\Theta(n \sqrt{\lg n})$
buckets; this will also enable us to `remove' the ceilings and floors
in the bound.  However, using a super-linear number of buckets uses
too much space if we use Lemma~\ref{lem:multis} to represent their
sizes.  Hence, we now develop a much more space-efficient alternative
to Lemma~\ref{lem:multis}, by giving more space-efficient FIDs.  In
particular, we show the following lemma which is an extension of
\cite[Proposition 4.3]{Pagh}.
%

\subsection{Fully Indexable Dictionaries for Dense Sets}

\begin{lemma}
\label{lem:full}
Given a set $S \subseteq [m]$, $|S| = n$, there is an FID on $S$ that 
requires ${\cal B}(n,m) + O(m \lg \lg m / \lg m)$ bits of space.
\end{lemma}
\begin{proof}
%
%
Take $u = \floor{ \frac{1}{2} \lg m}$ and 
divide the universe $[m]$ into $p = \ceil{m/u}$ {\it blocks\/}
of $u$ numbers each, with the $i$-th
block $U_i = \{ (i-1)u, \ldots, iu -1\}$, for $1 \leq i \leq p-1$, and
$U_p = \{ (p-1)u, \ldots, m-1 \}$.  Let $S_i = S \cap U_i$ and $n_i =
|S_i|$.   Clearly, we can view $S_i$ as a
subset of $[u]$, which we now do for convenience.
The set $S_i$ is represented implicitly by a string of
${{\cal B}(n_i, u)}$ bits by storing an index into a table
containing the characteristic bit vectors of all possible subsets of
size $n_i$ from a universe of size $u$.  $S$ is represented by
concatenating the representations of the $S_i$'s; the length of this
representation of $S$ is at most ${\cal B}(n,m) + O(m/\lg m)$ bits, as
shown in \cite{BM}.

To enable fast access to the representations of the $S_i$s, we store
two arrays of size $p$.  The first array $A$ stores the numbers $n_i$
in equal-sized fields of $\ceil{\lg u}$ bits each. The second array
$B$ stores the quantities ${\cal B}(n_i, u)$; since ${\cal B}(n_i, u)
\le u$ these numbers can also be stored in equal sized fields of
$\ceil{\lg u}$ bits each. This requires $O(m \lg \lg m/ \lg m)$ bits
of space.  We also store the prefix sums of the two arrays, as
described in \cite[Proposition 4.2]{Pagh} or \cite{Tarjan-Yao}, in
$O(m \lg \lg m/\lg m)$ bits, such that the $i$-th prefix sum is
calculated in $O(1)$ time.  We also store precomputed tables to
support $\Rank$ and $\Select$ queries on an arbitrary set $S_i$ given
its size and its implicit representation.  These tables require
$O(m^{1 - \epsilon})$ bits of space for some fixed positive constant
$\epsilon < 1$.

To find $\Rank(x)$ we proceed as in \cite{Pagh}: first compute $i =
\floor{x/u}$, find the number of elements in $S_0 \cup \ldots \cup
S_{i-1}$ using the partial sum data structure for the array $A$, index into
the string for $S$ to get the representation of $S_i$ using the
partial sum data structure for the array $B$, and find the rank of $x$
within the set $S_i$ using a table lookup.

To support \Select{} we do the following.   We let $v = \floor{(\lg p)^2}$
and $q = \floor{n/v}$. 
We store an array $C$ of size $q+2$ such that 
$C[0] = 0$, $C[q+1] = p$, and for $j = 1, \ldots, q$, $C[j]$
stores the index $l \le p$ such that $\Sigma_{i=1}^{l-1} n_i < jv \le \Sigma_{i=1}^{l} n_i$.  
The array $C$ takes $O(n/\lg p) = O(n/\lg m)$ bits and allows $\Select(j v)$ 
for $j = 1, \ldots, q$ to be answered in $O(1)$ time, as
follows.  Letting $k = C[j]$, we use the partial sums of $B$ to
extract the representation of $S_k$, use the partial sums of $A$ to
calculate $s = \sum_{i=1}^{k-1} n_i$, and use table lookup to
return the $(j v - s)$-th element from $S_k$ as the final answer.

To support $\Select$ for arbitrary positions, we follow the ideas of
\cite{clarkthesis,MRR}.  
For $i = 1,\ldots,q+1$, we
define the $i$-th {\it segment\/} as $\cup_{j=C[i-1]+1}^{C[i]} U_j$;
i.e., the part of the universe that lies between two successive indices
from $C$.  As $v > u$ for sufficiently large $m$, 
$C[i] > C[i-1]$ for all $1 \le i \le q$, and all 
segments (except perhaps the last) are nonempty.
We call a segment {\it dense\/} if its size is at most $(\lg p)^4$ and
{\it sparse\/} otherwise.  

For each sparse segment, we explicitly list
(in sorted order) the elements of $S$ that lie in that segment.  The
space required to represent the elements of $S$ that lie in a sparse
segment is therefore $O((\lg p)^2 \cdot \lg m)$, but since there are at
most $m/(\lg p)^4$ sparse segments, this adds up to $O(m / \lg m)$ bits
overall.
For a dense segment, we construct a complete tree with branching
factor $\ceil{\sqrt{\lg p}}$, whose leaves are the blocks that
constitute this segment.  Since the number of leaves is $O((\lg
p)^3)$, the depth of this tree is constant.  At each node of this
tree, we store an array containing the number of elements of $S$ in
each of its child subtrees.  If the tree for a dense segment has $k$
leaves, the space usage for this tree is $O(k \lg \lg p)$ bits.  As
segments are disjoint and the total number of blocks is $O(m/\lg m)$,
this adds up to $O(m \lg \lg m/\lg m)$ bits overall.
%
%
We store explicit pointers to the beginning of the representation of
each segment, which takes $O(m/\lg m)$ bits as there are only 
$O(m/(\lg m)^2)$ segments.

The representations of all sparse segments are stored consecutively,
as are the representations of all dense segments.
A bit-sequence of length $q+1$, where the $i$-th bit of the sequence
is 1 if the $i$-th segment is sparse and 0 otherwise, is used to
distinguish between the two cases; this bit sequence is stored
as a FID using Lemma~\ref{ranksel}.  Using $\Rank$ operations on 
this FID, we can access the representation of the $i$-th segment,
be it sparse or dense.

To compute $\Select(i)$ we first identify the segment in which the
$i$-th element can be found.  Letting $k_1 = C[\floor{i/v}]$,
by inspecting the prefix sums of $A$ at positions $k_1$ and $k_1 + 1$ 
one can determine whether the $i$-th element belongs to the segment 
ending at $k_1$ or the one beginning at $k_1 + 1$.  Suppose
it belongs to the segment $\sigma$.  Using the prefix sums of $A$, we
determine the rank of the element to be selected in $\sigma$.  If
$\sigma$ is sparse we read the required element directly from a sorted
array.  Otherwise, if $\sigma$ is dense, we start at the root of the
tree corresponding to $\sigma$ and do a predecessor search among the
numbers stored in the array stored at that node to find the subtree to
which the required element belongs. This can be done in constant time
via table lookup using tables of negligible size, as the array at each
node takes $O(\sqrt{\lg p} \lg \lg p) = o(\lg m)$ bits.  Thus, in
constant time we reach a leaf
that corresponds to some block $S_j$ which is known to contain the
element sought.  We find the number of elements $s$ in $S_0 \cup
\ldots \cup S_{j-1}$ using the partial sum data structure for the array
$A$, index into the string for $S$ to get the representation of $S_j$
using the prefix sum data structure for the array $B$, and find the position
$l$ of the $(i-s)$-th element in the representation of $S_i$ using a
table lookup. 

Now we consider supporting $\Rank$ and $\Select$ operations on ${\bar
S}$.  Again letting $\bar{S}_i = \bar{S} \cap U_i$ and $\bar{n}_i =
|\bar{S}_i|$, we observe that $\bar{n}_i = u - n_i$, and so the prefix
sums of $A$ suffice to answer prefix sum queries on the $\bar{n}_i$s.
Likewise, the implicit representation of $S_i$ is also an implicit
representation of $\bar{S}_i$ and the concatenated representations of
the $S_i$s is also an implicit representation of $\bar{S}$ that takes
only ${\cal B}(n,m) + O(m/\lg m)$ bits, from which the representation
of a single $\bar{S}_i$ can be retrieved in $O(1)$ time using the
array $B$.  Thus, answering $\Rank$ queries on $\bar{S}$ requires no
additional information except new tables (of negligible size) for
performing $\Rank$ and $\Select$ on the implicit representations of
the $\bar{S}_i$s.

To answer $\Select$ queries on $\bar{S}$, we create an array $\bar{C}$
which is analogous to the array $C$, and which partitions the universe
anew into segments.  Selecting elements from $\bar{S}$ in these
segments is done as before, with trees for dense segments and sorted
arrays for sparse segments.  This requires $O(m \lg \lg m/\lg m)$
additional auxiliary space.
%
\end{proof}
\begin{remark}
By replacing the implicit representations of the
$S_i$'s with the characteristic vector of set $S$, we get a
representation of a bit-vector of length $m$ that takes $m + O(m
\lg\lg m/ \lg m)$ bits and supports $\Rank_b$ and $\Select_b$ queries,
for $b \in \{ 0,1 \}$ (defined in Section \ref{sec:results}), in 
constant time. This improves the lower-order 
term in space of the earlier known data structures \cite{clarkthesis,MRR}
from $O(m/ \lg\lg m)$ to $O(m \lg\lg m/ \lg m)$.
\end{remark}

As an immediate consequence of Lemma~\ref{lem:full} we get:
\begin{corollary}
\label{cor:dense}
There is a fully indexable dictionary representation for a set $S
\subseteq [m]$, $|S| = n$ that uses ${\cal B}(n,m) + o(n)$ bits of
space, provided that $m$ is $O(n \sqrt{\lg n})$.
\end{corollary}

The following corollary is a consequence of Corollary \ref{cor:dense} and 
Lemma \ref{prefromset}. Note that ${\cal B}(n,m+n)$ is the information
theoretic minimum number of bits to represent a multiset of $n$ elements
from $[m]$.
\begin{corollary}
\label{cor:pref:dense}
If $m = O(n \sqrt{\lg n})$, then a sequence $S$ of $n$ non-negative 
numbers that sum up to 
$m$ can be represented using ${\cal B}(n,m+n) + o(n)$
bits to support $\Sum$ and $\Pred$ operations in constant time.
\end{corollary}

\subsection{Optimal Bucketing for Sparse Sets}

In this section we prove our main result.  A key idea will be to use
MSB bucketing to place keys into $\omega(n)$ buckets, and the following
proposition will be used to bound the increase in space usage:
%
\begin{proposition}\label{prop:increasem}
For all integers $x, y, c \ge 0$, $y \ge x$, ${\cal B}(x, y+c) - {\cal B}(x, y) = O(cx/y + \lg x + x^2/y)$.
\end{proposition}
\begin{proof}
We begin with the estimate 
${\cal B}(x,y) = x \lg(ey/x) - O(\lg x) - \Theta(x^2/y)$ \cite[Equation 1.1]{Pagh}.
From this it follows that 
${\cal B}(x,y+c) - {\cal B}(x,y) = O(x \lg ((y+c)/y) + \lg x + x^2/y)$ $ =
O(cx/y + \lg x + x^2/y)$. 
\end{proof}

Now we use Corollary~\ref{cor:pref:dense} to prove our main result:
\begin{theorem}
\label{main}
There is an indexable dictionary for a set $S \subseteq [m]$ of size
$n$ that uses at most $B(n,m) + o(n) + O(\lg \lg m)$ bits.
\end{theorem}
\begin{proof}
First, if $m < 4n \sqrt{\lg n}$ then we use Corollary \ref{cor:dense},
which establishes the result. If $m \ge 4n \sqrt{\lg n}$, we choose an
integer $l > 0$ such that $n \sqrt{\lg n} \le \lfloor m/2^{l} \rfloor
< 2 n \sqrt{\lg n}$.  We now group the keys based upon the mapping
$g(x) = \lfloor x/2^{l} \rfloor$.  Let $r = \lfloor (m-1)/2^l
\rfloor$.  We ``partition'' $S$ into sets $B_i$, for $i = 0,\ldots,r$,
where $B_i = \{ {{x} \bmod {2^{l}}} \mid x \in S \mbox{\rm \ and\ }
g(x) = i \}$.  Let $b_i = |B_i|$, for $i = 0,\ldots,r$.  We represent
the sequence $B_{top} = (b_0,\ldots,b_r)$ using the data structure of
Corollary \ref{cor:pref:dense} taking ${\cal B}(n, r+n+1) + o(n)$
bits, which supports $\Sum$ and $\Pred$ on $B_{top}$ in constant time.
By Proposition~\ref{prop:increasem}, 
${\cal B}(n,r+n+1) - {\cal B}(n,r) = O(n^2/r + \lg n) = o(n)$ as 
$r = \Theta(n \sqrt{\lg n})$.  Thus, the space usage is
${\cal B}(n,r) + o(n)$ bits.

The overall representation is the following. First we represent
$B_{top}$ as above. Then we represent each of the $B_i$'s using the
data structure of Lemma \ref{lem:multiple}. The total space used will be $n
l + {\cal B}(n, r) + o(n) + O(\lg \lg m)$ bits. Note that ${\cal B}(n,
r) = n \lg (e r/n) + o(n)$ as $r = \Theta(n \sqrt{\lg n})$, and so $n
l + {\cal B}(n, r) = nl + n \lg (me/(2^{l}n)) + o(n) = {\cal B}(n,m) +
o(n)$.  Thus, the overall space bound is as claimed.  The computations
of $\Rank{}$ and $\Select{}$ proceed essentially as in
Theorem~\ref{single}, except that we use
Corollary~\ref{cor:pref:dense}, instead of Lemma~\ref{lem:multis}, to
represent $B_{top}$.  
\end{proof}

\section{An indexable dictionary in the cell probe model}
\label{sec:cellprobe}
In this section we give an indexable dictionary representation 
for a set $S$ of size $n$ from a universe of size $m$ that uses ${\cal
B}(n,m) + o(n)$ bits of space in the cell probe model \cite{Yao}. 
Recall that in this model, time
is measured as just the number of words (cells) accessed during an
operation. All other computations are free.  We first prove a lemma
which is analogous to Lemma \ref{simple}, but which does not 
assume access to the functions $h_S$ and $q_S$.
\begin{lemma}
\label{lem:dist}
There is an indexable dictionary for a set $S \subseteq [m]$ of size
$n$ that uses $n (\lg m + \lg n + O(1))$ bits in the cell probe model.
\end{lemma}
\begin{proof}
  Let $x_1 < x_2 < \dots < x_n$ be the elements of $S$.
  
  If $n \ge \sqrt{\lg m}$, then we first store the given set $S$ in an
  array $A$ in increasing order, which takes $n \lg m + O(n)$ bits of
  space. As in Lemma \ref{simple}, we find a minimal perfect hash
  function $f$ for $S$ and store it using $O(n + \lg \lg m) = O(n)$
  bits (since $n \ge \sqrt{\lg m}$). We then store a table $T$ with
  $T[f(x_i)] = i$. This requires $n \lg n + O(n)$ bits. To answer
  $\Rank(x)$, we calculate $j = T[f(x)]$ and check if $x = x_j$; if so
  we return $j-1$, otherwise return $-1$. Supporting $\Select$ is
  straightforward, as we store the elements in sorted order in $A$.
  
  Otherwise, if $n < \sqrt{\lg m}$, let $s = \floor{\frac{\lg m}{n^2}}$ and
  $r = \ceil{(\lg m)/s}$, and note that $r = O(n^2)$. We divide the 
  $\ceil{\lg m}$-bit
  representation of each $x\in S$ into $r$ contiguous pieces, 
  where each piece has
  size exactly $s$ bits, except for one piece (consisting,
  say, of the most significant bits of $x$) which has size $s'$
  bits, $1 \le s' \le s$.  We number the parts $0,\ldots,r-1$ with
  $0$ being the most significant.  Since $n < \sqrt{\lg m}$, $s \ge
  1$, and this is possible.  Then there exists a set $R \subseteq
  [r]$, $|R| = n$, such that if we consider only the bits in the
  parts that belong to $R$, all keys in $S$ are still distinct
  \cite{AFK}.  Let $h(x,R)$ be the number obtained by
  extracting the bits in $x$'s representation from parts that belong
  to $R$, and concatenating them from most significant to least
  significant.  Then for any distinct $x, y \in S$, $h(x,R) \not =
  h(y,R)$.  Similarly let $q(x,R)$ be the number obtained by
  extracting the bits in $x$'s representation from parts that do {\it
  not\/} belong to $R$, and concatenating them from most significant
  to least significant.\footnote{It appears %
  to be difficult to compute %
  $h(x,R)$ and $q(x,R)$ %
  in $O(1)$ time on the RAM model.}  
  The set $S$ is represented as follows.

  First, we store an implicit representation of $R$; this takes
  $\ceil{\lg {{r}\choose{n}}} = n \lg n + O(n)$ bits.
   Then, we store the sequences $h(x_1,R), h(x_2,R), \ldots, h(x_n,R)$
   and  $q(x_1,R),$ $q(x_2,R), \ldots, q(x_n,R)$ in that order.
   Clearly, this representation takes $n \lg m + n \lg n + O(n)$ bits.
  
  
  To answer {\Rank}($x$), we read $R$ first; as $n \lg n = o(\lg m)$ 
  this can be done in $O(1)$ time.  Then we compute $h(x,R)$ in $O(1)$ time.
  We then read $h(x_1,R),\ldots, h(x_n,R)$; since $h(x_i,R)$ is $O((\lg m)/n)$
  bits long, all these values can be read in $O(1)$ time.   We then
  find an $i$ such that $h(x_i,R) = h(x,R)$; if such an $i$ exists,
  we verify the match by reading $q(x_i,R)$ and comparing it with $q(x,R)$,
  and return $i-1$ or $-1$ as appropriate.  If such an $i$ does not 
  exist then $x \not \in S$.  It is easy to see that $\Select$ can also be
  supported in constant time using this representation.  
\end{proof}

Using the representation of Lemma \ref{lem:dist} instead of Lemma
\ref{simple} for representing the sets at the bottom level in the
proof of Theorem \ref{single}, we get an indexable dictionary
data structure that takes $n \ceil{\lg m} + o(n)$ bits. One can use this
data structure to get a result similar to Lemma \ref{lem:multiple}, 
but without
the additive $O(\lg\lg m)$ term in the space complexity. Thus we have:
\begin{lemma}
\label{lem:multiple-cp}
Let $S_1, S_2 , \ldots, S_s$ all contained in $[m]$ be given sets with
$S_i$ containing $n_i$ elements, such that $\sum_{i=1}^s n_i =n$. Then
this collection of sets can be represented using $n \ceil {\lg m } +
o(n)$ bits where the operations $\Rank (x, S_i)$ and $\Select (j,
S_i)$ can be supported in constant time in the cell probe model, for
any $x \in [m], 1 \leq j \leq n$ and $1\leq i \leq s$.  This requires
that we have access to a constant-time oracle which returns the prefix
sums of the $n_i$ values.
\end{lemma}

Using Lemma~\ref{lem:multiple-cp} in place of Lemma~\ref{lem:multiple} in 
Theorem \ref{main}, we get the following result for the cell probe model:
%
\begin{theorem}
\label{cell}
There is an indexable dictionary for a set $S \subseteq [m]$ of size
$n$ using ${\cal B}(n,m) + o(n)$ bits in the cell probe model.
\end{theorem}
%

%

As an immediate corollary we get:
\begin{corollary}
  There is an indexable dictionary for a set $S \subseteq [m]$ of size
  $n$ using at most $n \ceil{\lg m}$ bits in the cell probe model.
\end{corollary}

\section{Extensions and applications} 

In this section, we give some extensions and applications of our
succinct indexable dictionary (Theorem \ref{main}) as well as our
fully indexable dictionary for dense sets (Corollary \ref{cor:dense}).


\subsection{Multiple Indexable Dictionaries}
Here, using our succinct indexable dictionary, we will give a better
representation for multiple indexable dictionaries, improving on Lemma
\ref{lem:multiple}.

Let $S_0, S_1 , \ldots, S_{s-1}$ all contained in $[m]$ be a given
sequence of dictionaries with $S_i$ containing $n_i$
elements, such that $\sum_{i=0}^{s-1} n_i =n$.  Note that the
representation in Lemma~\ref{lem:multiple} of these multiple
dictionaries requires an oracle to specify the starting point of each
dictionary in the sequence.
The representation we develop here does not make use of this
assumption, but instead requires that $s = O(n)$.
Define the set $S$ as follows:
$$S = \{ \langle i,j \rangle : i \in [s], j \in [m] {\rm \ and } \ j
\in S_i \}.$$
We map the pairs $\langle i, j \rangle, i \in [s], j \in [m]$ to
integers in the range $[ms]$ using the obvious mapping $\langle i, j
\rangle \mapsto i \cdot m + j$.  We represent the $n$-element set
$S$ using our indexable dictionary representation of Theorem
\ref{main}, which takes ${\cal B}(n,ms) + o(n) + O(\lg \lg ms)$ bits.
As any $n$-element subset of $[ms]$ corresponds to a unique sequence
of $s$ sets (using the inverse of the above mapping), the first term
${\cal B}(n,ms)$ is the minimum number of bits required to represent
such a sequence of multiple dictionaries.

Now to support the multiple dictionary operations $\Rank (x, S_i)$ and
$\Select (j, S_i)$, we need to find the rank of $\langle i, 0 \rangle$
in $S$ even if $\langle i, 0 \rangle \not \in S$.  We can do this by a
more detailed inspection of the proof of Theorem~\ref{main}, and
potentially modifying $S$ slightly.

If $m s \le 4 n \sqrt{\lg n}$, then the set $S$ is dense and so this
follows from Lemma~\ref{lem:full}.  If $m s > 4 n \sqrt{\lg n}$ then
we alter $m$ to a new and carefully-chosen value $m'$, and redefine
$S$ with the new value of $m'$; more precisely the pairs in $S$ stay
the same, but we change the mapping that takes pairs to integers as
$\langle i, j \rangle \mapsto i \cdot m' + j$ .  By doing this, we
ensure that no bucket at the top level of Theorem~\ref{main} contains
elements of the form $\langle x, y \rangle$ and $\langle x', y'
\rangle$ for $x \ne x'$ (i.e., all elements in a bucket have the same
first co-ordinate). Thus, answering rank queries for $\langle x, 0
\rangle$ only requires summing up the sizes of a number of top-level
buckets, which is supported by the top level representation.

We now discuss the choice of $m'$.  Recall that if we apply
Theorem~4.1 directly to $S$, we would choose an integer $l$ such that
$n \sqrt{\lg n} \le \lfloor{ms/2^l} \rfloor < 2n \sqrt{\lg n}$ and
place $x$ in the bucket $\lfloor{x/2^l}\rfloor$.  Let $l$ be this
integer, and let $m' = 2^l \cdot \lceil{m/2^l}\,\rceil$, i.e., round
the value of $m$ to the next higher multiple of $2^l$.  Now it is easy
to verify that $\lfloor{(x \cdot m' + y)/2^l}\rfloor \ne \lfloor{(x'
\cdot m' + y')/2^l}\rfloor$ for $x \ne x'$, and thus keys belonging to
distinct dictionaries are mapped to different buckets.

However, this increases the universe size to $m's$ from $ms$.  Due to
this increase, a direct application of Theorem~\ref{main} may result
in the elements being bucketed according to the mapping $x \mapsto
\lfloor{x/2^{l'}}\rfloor$, for some $l' \ge l$.  This issue is most
easily dealt with by noting that as $m' < m + 2^l$, $m's < ms(1 +
2^l/m) = ms(1 + \Theta(s/(n \sqrt{\lg n}))) = ms(1 + O(1/\sqrt{\lg
n}))$ (recall that $s = O(n)$ by assumption).  This in particular
means that, for $n$ larger than some constant, $m's < 2ms$, and so
retaining the mapping $x \mapsto \lfloor{x/2^l}\rfloor$ in the proof
of Theorem~\ref{main} gives at most $4 n \sqrt{\lg n}$ buckets at the
top level, which is immaterial.  More importantly, since $m's = ms (1
+ O(1/\sqrt{\lg n}))$, the increase in the space is only in the
lower-order terms by Proposition~\ref{prop:increasem}.  
With this additional power, we now support the
multiple dictionary operations as follows:
\begin{itemize}
\item To find the size of the set $S_i$, we do the following.  Find
  the rank of $\langle i+1, 0\rangle$ and the rank of $\langle
  i,0\rangle$.  The difference gives the size of the set $S_i$.
  
\item To perform $\Select(i, S_j)$, find the rank $r$ of $\langle j,0
  \rangle$ and then do $\Select (r+i)$ in $S$. The second coordinate
  of the element returned by the $\Select$ operation is the value of
  the $i$-th smallest element of $S_j$.
  
\item To find $\Rank(x, S_j)$, find and subtract the rank of $\langle
  j,0\rangle$ from $\Rank (\langle j,x \rangle)$.  Return the result
  if $\Rank (\langle j,x \rangle) \geq 0$ and return $-1$ otherwise.
\end{itemize}

Thus we have:
\begin{theorem}\label{th:multiple}
  Let $S_0, S_1 , \ldots, S_{s-1}$ all contained in $[m]$ be a given
  sequence of $s =O(n)$ sets with $S_i$ containing $n_i$ elements,
  such that $\sum_{i=1}^s n_i =n$. Then this collection of sets can be
  represented using ${\cal B}(n,ms) + o(n) + O(\lg \lg m)$ bits and
  the $\Rank (x, S_i)$ and $\Select (j, S_i)$ operations can be
  supported in constant time for any $x \in [m], i \in [s]$ and 
  $j \in \{1,\ldots,n_i\}$. 
  We can also find $n_i$ for each $i$ in constant time.
  The first term in the space bound is the minimum number of bits
  required to represent such a sequence of sets.
\end{theorem}

We use Theorem~\ref{th:multiple} in the next section to represent
$k$-ary trees.  However, Theorem~\ref{th:multiple} has several direct
applications.  For instance, it can be used to represent an arbitrary
directed graph on $n$ nodes, where the vertices
are numbered $0$ to $n-1$ and $S_i \subseteq [n]$ represents 
the set of neighbors of vertex $i$.  The space used --- 
${\cal B}(r, n^2) + o(r)$ bits, where $r$ is the number of edges --- is 
information-theoretically optimal, and the representation supports
the union of the operations supported in $O(1)$ time by the
standard adjacency list and adjacency matrix representations,
such as adjacency testing, or iteration over the list of neighbors 
of a given vertex.  However, it also supports constant-time operations
not supported in $O(1)$ time by either of the standard 
representations, including random access to the $i$-th 
neighbor of a vertex and reporting the out-degree of a vertex.

\subsection{Representing $k$-ary Cardinal Trees }
\label{sec:kary}

Recall that a $k$-ary cardinal tree is a rooted tree, 
each node of which has $k$
positions labeled $0, \ldots, k-1$, which can contain edges to
children.  As noted in the introduction, 
the space lower bound for representing a $k$-ary cardinal
tree with $n$ nodes is ${\cal C}(n,k) = \left \lceil \lg \left (
\frac{1}{kn+1}{{kn + 1} \choose n} \right ) \right \rceil$.
%
%
We now give a succinct representation of $k$-ary 
cardinal trees that supports a number of operations in $O(1)$ time.
Given a node, we can go to
its child labelled $j$ (i.e. the child reachable 
with an edge in the position labelled $j$), its 
$i$-th child or to its parent if these nodes exist. 
In addition, we can determine the
degree of a node as well as the ordinal position of a node among
its siblings in constant time.  The representation uses
${\cal C}(n,k) + o(n) + O(\lg \lg k)$ bits of space; the space
usage is therefore information-theoretically optimal up
to $o(n+\lg k)$ terms, and is more space-efficient than the
representation of \cite{BDMRRR}.
Unfortunately, we are not able to support the
subtree size operation in constant time using this representation.
Our representation imposes a numbering from $0$ to $n-1$ on the nodes
(the representation of \cite{BDMRRR} also imposes a numbering, albeit
a different one, on the nodes).
\begin{theorem}\label{karyopt}
A $k$-ary tree on $n$ nodes can be represented using ${\cal C}(n,k) +
o(n) + O(\lg \lg k)$ bits where given a node of the tree, we can go to
its $i$-th child or to its child labeled $j$ or to its parent if they
exist, all in constant time.  In addition, we can determine the
degree of a node as well as the ordinal position of a node among
its siblings in constant time.
\end{theorem}
\begin{proof}
  Consider a level-ordered left-to-right numbering of the tree nodes
  by numbers from $\{0,\ldots,n-1\}$, starting from the root with $0$.
%
%
  From now on, we refer to the nodes of the tree by these numbers. By a
  child labeled $j$ of a node $x$, we mean the child $y$ of $x$ such
  that the edge $(x,y)$ is labeled $j$.  Let $S_x$ be the set of edge
  labels out of the vertex $x$.  Then the sets $S_0, \ldots, S_{n-1}$
  form a sequence of $n$ sets of total size $n-1$, each being a subset
  of $[k]$.
  
  Representing these multiple dictionaries using 
  Theorem \ref{th:multiple}, 
   we get a
  representation for the $k$-ary tree using at most 
  ${\cal B}(n-1,kn) + o(n) + O(\lg \lg (kn))$ bits. Since ${{kn}\choose{n-1}} 
   = \frac{n}{kn+1} {{kn+1}\choose{n}}$, ${\cal B}(n-1,kn) + o(n) + O(\lg \lg (kn))= {\cal B}(n,kn+1) - \lg(kn + 1) + o(n) +
  O(\lg \lg kn) = {\cal C}(n, k) + o(n) + O(\lg \lg k)$ bits. 
By
  Theorem \ref{th:multiple}, we can support the degree of a node $x$,
  the $i$-th child of a node $x$, and the ordinal position (the local
  rank) of the child labeled $j$, if exists, of a node $x$, all in
  constant time.  However, the basic navigational operations of going
  to a child or to the parent are not supported.  To support these, we
  re-examine the proof of Theorem~\ref{th:multiple}.  Note that in
  applying Theorem~\ref{th:multiple} to represent our tree, the
  following set $S$ is stored in an indexable dictionary:
  $$S = \{ \langle x,j \rangle : x \in [n], j \in [k] \mbox{\rm\ and
  $\exists$ an edge labeled $j$ out of node} \mbox{\rm\ $x$}\}.$$
  The representation supports $\Rank (\langle x,j\rangle, S)$ and
  $\Select (\langle x,j\rangle, S)$ in $O(1)$ time.  It is easy to
  verify that:
\begin{itemize}
\item $\Rank (\langle x,j\rangle,S) + 1$ gives the label of the child
  labeled $j$ of node $x$, if it exists, and returns $0$ otherwise.
\item The first component of $\Select (i,S)$ is the parent of the node
  $i$. I.e., if the $i$-th element in $S$ is $\langle
  x,j\rangle$, then $x$ is the parent of the node $i$, for $i>0$.
\end{itemize}
\end{proof}

%

\subsection{Multisets}

Given a multiset $M$ from $U=[m]$, $|M| = n$, an {\it indexable
  multiset\/} representation for $M$ must support the following two 
operations in constant time:

\begin{description}  
\item[$\Rankm(x,M)$] Given $x \in U$, return $-1$ if $x \not \in M$ and
  $|\{y \in M | y < x \}|$ otherwise, and
  
\item[$\Selectm(i,M)$] Given $i \in \{1,\dots,n\}$, return the largest
  element $x \in M$ such that $\Rankm(x) \le i - 1$.
\end{description}
We also consider a generalization of 
the $\Rankm$ operation:

\begin{description}
  
\item[$\Rankmplus(x)$] Given $x \in U$, return $|\{y \in M | y < x
  \}|$.

\end{description}

There is an intimate connection between FIDs and multisets similar
to that in Lemma \ref{prefromset} as shown below.
\begin{lemma}
\label{multfromset}
Suppose there is an FID representation for any given set $T \subseteq
U$ using $f(|T|,|U|)$ bits of space.  Then given a multiset $M$ of $n$
elements from the universe $[m]$, there is a data structure to represent
$M$ using $f(n,m+n)$ bits of space that supports $\Rankmplus$ and
$\Selectm$ operations in constant time.
\end{lemma}
\begin{proof}
  Consider the $m+n$ bit representation of $M$ obtained as follows.
  For $i=0$ to $m-1$, represent $i$ by a $1$ followed by $n_i$ $0$s
  where $n_i$ is the number of copies of the element $i$ present in
  the set $M$.  Clearly this representation takes $m+n$ bits since it
  has $m$ $1$s and $n$ $0$s.
  
  View this bit sequence as a characteristic vector of a set $T$ of
  $m$ elements from the universe $[m+n]$.  Represent $T$ as an FID
  using $f(m,m+n) = f(n,m+n)$ bits. It is easy to verify that 
  $\Rankmplus(x,M) = \Select(x, T) - x +1$ and $\Selectm(i,M) =
  \Select(i, \bar{T}) - i$. (Note that $[m+n]$ starts with element
  $0$.)  
\end{proof}

The following corollary follows from Corollary \ref{cor:dense} and
Lemma \ref{multfromset}.
\begin{corollary}
\label{cor:multi:dense}
Given a multiset $M$ of $n$ elements from the universe $[m]$, there is
a data structure to represent $M$ and to support $\Rankmplus$ and $\Selectm$
operations in constant time using ${\cal B}(n,m+n) + o(n)$ bits of
space, provided that $m = O(n \sqrt{\lg n})$.
\end{corollary}

We now develop an indexable multiset representation (that
supports only $\Rankm$ and $\Selectm$ operations) taking ${\cal
  B}(n,m+n) + o(n) + O(\lg \lg m)$ bits for all $n$.  As was alluded
to in the introduction (see \cite{elias}), the first term is the
minimum number of bits required to store such a multiset.
%
%
%
%
%
\begin{theorem}
\label{thm:multiset}
Given a multiset $M$ of $n$ elements from $[m]$, there is an indexable
multiset representation of $M$ that uses ${\cal B}(n, m+n ) + o(n) +
O(\lg \lg m)$ bits.
\end{theorem}
\begin{proof}
  If $n$ is dense in $m$, i.e. if $m = O( n \sqrt { \lg n })$ then
  the theorem follows from Corollary \ref{cor:multi:dense}.
  
  If not, then we represent $M$ as follows.  First represent the set
  $S$ of distinct elements present in $M$ using the indexable
  dictionary data structure of Theorem \ref{main} using ${\cal B}(n', m) +
  o(n) + O(\lg \lg m)$ bits where $n' \leq n$ is the number of
  distinct elements present in $M$.

  Then represent the rank information separately by representing each
  element $i$ present in $M$ (in increasing order) by a $1$ followed
  by $n_i - 1$ $0$s where $n_i$ is the multiplicity of the element $i$
  in $M$.  This representation is a bitstring of length $n$ with $n'$
  $1$'s.  This bitstring could be considered as a characteristic
  vector of a set $R \subseteq [n]$ with $|R|=n'$. Let $\bar R = [n]
  \setminus R$.
  
  Now to find $\Rankm(x,M)$, first find $\Rank (x, S)$. If the answer is
  $-1$, then return $-1$. Otherwise $\Rankm(x,M)$ is $\Select
  (\Rank(x,S) + 1, R)$.  To find $\Selectm (i,M)$, let $r = \Rank (i, R)
  + 1$ if $\Rank (i, R) \geq 0$ and $r = i - \Rank (i, \bar R)$
  otherwise. The value $r$ is precisely the number of $1$'s up to and
  including $i$ in the characteristic vector of $R$.  Then $\Selectm
  (i,M) = \Select (r, S)$.

  To support both $\Rankm(x,M)$ and $\Selectm(i,M)$ in constant time
  in this way, we need a fully indexable dictionary for $R$.
  If $n'$ is dense in $n$, i.e. $n = O(n' \sqrt {\lg n' } )$, then use
  the fully indexable dictionary of Corollary~\ref{cor:dense} for $R$.
  This uses ${\cal B}(n', n) + o(n')$ bits for a total of ${\cal
    B}(n', m) + {\cal B}(n', n) + o(n) + O(\lg \lg m)$ including the
  space for representing $S$. Clearly this space is ${\cal B}(n', m) +
  {\cal B}(n-n', n) + o(n) + O(\lg \lg m)$ which is at most ${\cal
    B}(n, m+n) + o(n) + O(\lg \lg m)$.
  
  Otherwise represent $R$ using the FID representation of
  Lemma~\ref{ranksel} which uses $n + o(n)$ bits.  Since $n'$ is
  sparse in $n$, $n'< c n/\sqrt {\log n}$ (for some constant $c$) in
  which case ${\cal B}(n',m) + n \leq {\cal B}(n,m) + o(n)$.  To see this,
  note that ${m \choose n} = \frac{m-n+1}{n} {m \choose {n-1}} \geq 2
  {{m}\choose{n-1}}$ since $(m-n+1)/n > 2$ for sufficiently large $m$
  and $n \leq d m /\sqrt {\lg m}$ for some constant $d$.  Hence ${\cal
    B}(n,m) \geq {\cal B}(n-1, m) +1$ and so ${\cal B}(n,m) \geq {\cal
    B}(n',m) + n - n'$. That is, ${\cal B}(n',m) + n \leq {\cal
    B}(n,m) + n' \leq {\cal B}(n,m) + o(n)$.  Finally, 
${\cal B}(n,m) + o(n) \leq {\cal B}(n, m+n) + o(n)$, by Proposition~\ref{prop:increasem} as $n$ is also sparse in $m$.  
\end{proof}

\subsection{Applications of Succinct FIDs}

Here we give more applications of our FID for
dense sets obtained in Corollary \ref{cor:dense} to represent
sets, multisets and prefix sums.  

\subsubsection{Set data structure with $\Select$}

\begin{theorem}
\label{thm:dictsele}
There is a representation of a set $S \subseteq [m]$ of size $n$ that
uses at most ${\cal B}(n,m) + o(n)$ bits and supports the $\Select$
operation in constant time.
\end{theorem}
\begin{proof}
  The proof is essentially as in the proof of Theorem \ref{main}
  except that we store each of the $B_i$'s as a sorted list.
  
  To perform a $\Select (i)$ operation, we first perform a $\Pred(i)$
  at the top level prefix sum representation, to find the bucket $B_j$
  in which the $i$-th element is present.  Then a $\Sum (j)$ operation
  at the top level representation gives the prefix sum of the first
  $j-1$ bucket sizes. Now $i - \Sum (j)$ is the rank of the element in
  the bucket $B_j$ in which we are interested.  Since the buckets
  are sorted, it is easy to find the element of appropriate rank in
  that bucket.  
\end{proof}

%

%
%

\subsubsection{Multiset with $\Selectm$ and Prefix Sums}

\begin{theorem}
\label{thm:impmultiset}
Given a multiset $M$ of $n$ elements from $[m]$, there is a
representation of $M$ that uses ${\cal B}(n, m+n ) + o(n)$ bits that
supports $\Selectm$ operation in constant time.
\end{theorem}
\begin{proof}
  Use the encoding given in the proof of Lemma \ref{multfromset} to
  convert the multiset $M$ into a set $T \subseteq [m+n]$ of size $n$.
  Represent $T$ using the representation of Theorem \ref{thm:dictsele}
  which uses ${\cal B}(n, m+n) + o(n)$ bits and supports $\Select$
  operation on $T$. From the proof of Lemma \ref{multfromset}, we know
  that $\Selectm (i, M) = \Select (i, T) -i$.  The theorem follows.
\end{proof}

As an immediate corollary we get:
\begin{corollary}
\label{thm:prefimp2}
Given a sequence $X = x_1, \ldots , x_n$ of non-negative integers such
that $\sum_{i=1}^n x_i = m$, the sequence can be represented using
${\cal B}(n, m+n) + o(n)$ bits to support the partial sum query
$\Sum(i,X)$ in constant time.  The first term is the information
theoretically minimum number of bits required to represent such a
sequence.
\end{corollary}
\begin{proof}
  Consider the multiset $M$ of partial sum values $M = \{
  \sum_{j=1}^{i} x_j: 1 \leq i \leq n \}$.  As the $x_i$s are
  non-negative and add up to $m$, $M \subseteq [m]$.  Represent this
  multiset using Theorem~\ref{thm:impmultiset} and observe that
  $\Sum(i,X) = \Selectm(i,M)$.  Also, as the mapping from $X$ to $M$
  is invertible, the information theoretic minimum number of bits
  required to store the partial sum information is ${\cal B}(n, m+n)$.
  The result follows.
%
%
%
\end{proof}

\section{Optimality considerations}
\label{opt}
As mentioned in the introduction, some of the space bounds we show
above may actually be very far from the information-theoretically
optimal bound of ${\cal B}(n,m)$.  Recall that, for
example, in the context of storing a set of size $n$ from $[m]$, the
information-theoretic lower bound of ${\cal B}(n,m)$ bits may be
dwarfed by additive terms of $o(n)$ bits, say when $n = m-c$ for some
constant $c$.  We now note that this is unavoidable to some extent,
and in particular that achieving a space bound even polynomial in the
information-theoretic lower bound and preserving constant query time
is impossible for several of the problems that we consider in this
paper, namely:
\begin{itemize}
\item Supporting $\Rank$ queries on a set of integers;
\item Supporting $\Select$ queries on a set of integers and
\item Supporting $\Sum$ queries on a sequence of non-negative integers
(or equivalently, supporting $\Selectm$ on a multiset of integers).
\end{itemize}
We show that the $\fullrank$ problem reduces to all of these.  
Given a set $S$ of size $n$ from $U = [m]$, recall that
$\fullrank(x,S)$ returns the rank of $x$ in $S$ for any $x \in U$. 
Beame and Fich \cite[Corollary 3.10]{BeameFich} showed the following:
\begin{lemma}
\label{lem:beamefich}
Given a set $S \subseteq [m]$, $|S| = n$, any data structure that uses 
$n^{O(1)}$ words of space would require 
$\Omega(\sqrt{\lg n/\lg \lg n})$ time in the worst case to answer
 $\fullrank$ queries in the cell probe model with word
size $(\lg m)^{O(1)}$.
\end{lemma}
%
%
\begin{lemma}
  Given a set $S \subseteq [m], |S| = n$, one cannot support $\Rank$
  queries on $S$ in $O(1)$ time for all $n$ using $({\cal
  B}(n,m))^{O(1)}$ bits in the cell probe model with word size
  $(\lg m)^{O(1)}$ bits.
\end{lemma}
\begin{proof}
  Suppose that the statement of the lemma is
  false.  Then we would solve $\fullrank$ queries on $S$ in $O(1)$
  time using $n^{O(1)}$ words of space as follows, contradicting
  Lemma~\ref{lem:beamefich}.
  Letting ${\cal B} = {\cal B}(n,m)$, by assumption we can store $S$
  in ${\cal B}^{O(1)}$ bits, which is $(n \lg m)^{O(1)}$ bits of space, and
  answer $\Rank$ queries on $S$ in $O(1)$ time. Similarly, we can
  store $\bar{S}$ in ${\cal B}(m-n,m) = n^{O(1)}$ words of space and
  answer $\Rank$ queries on $\bar{S}$ in $O(1)$ time.  To answer
  $\fullrank(x,S)$ for all $x \in U$, we first compute $\Rank(x,S)$;
  if the value returned is not $-1$ we return it as the answer to
  $\fullrank$, otherwise we return $x - \Rank(x,\bar{S}) - 1$ as the
  answer to $\fullrank$.
\end{proof}
\begin{lemma}
  Given a set $S \subseteq [m], |S| = n$, one cannot support $\Select$
  queries in $O(1)$ time for all $n$ using $({\cal B}(n,m))^{O(1)}$
  bits in the cell probe model with word size $(\lg m)^{O(1)}$
  bits.
\end{lemma}
\begin{proof}
  Suppose that the statement of the lemma is false.  
  Then given $T \subseteq [m^*]$, where $T= \{ t_1, \ldots, t_{n^*} \}$ 
  and $t_1 < t_2 < \ldots < t_{n^*}$, we can answer $\fullrank$ queries 
  on $T$ in $O(1)$ time using $(n^*)^{O(1)}$ words as follows, contradicting 
  Lemma~\ref{lem:beamefich}.  We create a 
  bit-vector by writing down $t_1$ $0$s
  followed by a $1$, then $(t_2 - t_1)$ $0$s followed by a $1$, and so
  on and finally we write $(m^* - t_n)$ $0$s.  This is a bit vector
  with $n^*$ $1$s and $m^*$ $0$s; we view this is as a characteristic
  vector of a set $S$ of size $n = n^*$ from $[m]$ where $m = m^* +
  n^*$.  We store $\bar{S}$ using $({\cal B}(n,m))^{O(1)}$ bits $=
  n^{O(1)}$ words and compute $\fullrank(x,T)$ as $\Select(x, \bar{S})
  - x$ in $O(1)$ time.
%
%
\end{proof}
\begin{lemma}
   Given a sequence $X = x_1, \ldots, x_n$ of non-negative
  integers adding up to $m$, one cannot store this sequence in $({\cal
    B}(n,m+n))^{O(1)}$ bits of space for all $m, n$ and support the
  $\Sum$ query on this sequence in $O(1)$ time in the cell probe model
  with word size $(\lg(m+n))^{O(1)}$ bits.
\end{lemma}
\begin{proof}
  Suppose that the statement of the lemma is
  false.  Then given $S = \{ s_1, \ldots, s_{n^*}\} \subseteq [m^*]$
  where $s_1 < s_2 < \ldots < s_{n^*}$ we could answer $\fullrank$
  queries on $S$ in $O(1)$ time using ${(n^*)}^{O(1)}$ words of space 
  as follows, contradicting Lemma~\ref{lem:beamefich}.
  Create a sequence $X$ which consists of $s_1$ $0$s
  followed by a $1$, and for $i=2$ to $n$, $s_i - s_{i-1}-1$ $0$s
  followed by a $1$ and finally $m^* - s_{n^*} -1$ $0$s followed by a
  $1$. This is a sequence of $n = m^*+1$ non-negative integers adding
  to $m = n^{*}$.  We store this sequence using $({\cal
    B}(n,m+n))^{O(1)}$ bits, but since ${\cal B}(n,m+n) = {\cal
    B}(m,m+n) = {\cal B}(n^*,m^*+n^*+1) \le {\cal B}(n^*,2m^*+1)$, the
  space usage is $(n^*)^{O(1)}$ words.  It is easy to verify that
  $\fullrank(j,S) = \Sum(j,X)$. 
\end{proof}

\section{Conclusions}

We have given a static data structure 
for storing an $n$ element
subset of an $m$ element universe, that takes 
${\cal B}(n,m) + o(n) + O(\lg \lg m)$ bits of space
and supports $\Rank$ and $\Select$ operations in constant time 
(an indexable dictionary) on the RAM model of computation. 
${\cal B}(n,m)$ is the information theoretically optimal
number of bits needed to store a subset of size $n$ from an $m$-element
universe. By modifying our indexable dictionary for the RAM model,  
we obtained an 
indexable dictionary representation that uses ${\cal B}(n,m) + o(n)$ bits 
in the cell probe model.
This, in particular, implies that $n$ words (of size $\ceil{\lg m}$ bits) 
are sufficient to represent $n$ elements from an $m$ element 
universe and answer membership queries
in $O(1)$ time on the cell probe model, answering
a question raised by Fich and Miltersen \cite{FM}
and Pagh \cite{Pagh}.

Using the indexable dictionary representation for the RAM model, 
we have developed improved
succinct representations for a number of objects.
We have shown that a $k$-ary tree on $n$ nodes
can be represented using ${\cal C}(n,k) + o(n) + O(\lg \lg k)$ bits of
space and support all the navigational operations, except the subtree
size of a given node, in constant time. Here ${\cal C}(n,k)$ is the
information-theoretically optimum number of bits required to represent
a $k$-ary tree on $n$ nodes.  We also developed a succinct 
representation for an indexable multiset of $n$
elements from an $m$ element universe using ${\cal B}(n,m+n) + o(n) +
O(\lg \lg m)$ bits.

An important subroutine used by the indexable dictionary is
a space-efficient fully-indexable dictionary (FID)
which simultaneously supports $\Rank$ and $\Select$ on a set and
its complement.  This data structure, which is functionally equivalent
to supporting $\Rank_{0/1}$ and $\Select_{0/1}$ on a bit-vector,
occupies ${\cal B}(n,m) + o(m)$ bits and supports all operations in
$O(1)$ time on the RAM model.  We gave further applications of 
this result, most notably, to representing a sequence of non-negative
integers in information-theoretically optimal space, while supporting
prefix sum queries in $O(1)$ time.

We have focussed on space utilization and (static) query time, and have not
given extensive consideration either to the time and space required
for pre-processing, or to dynamizing the data structure.  As regards
the pre-processing time, we note that an indexable dictionary can be
used to sort the input set $S$, so pre-processing must take at least
as much time as the best algorithm for sorting integers.  Assuming
$S$ is presented in sorted order, however, the main bottleneck is the creation
of hash functions as required by Lemma~\ref{lem:sharing}.  
The hash functions can be found rapidly using randomization, yielding
a linear expected time pre-processing algorithm.  As regards dynamization,
the work of \cite{FS89} gives a lower bound of $\Omega(\lg n/\lg \lg n)$ 
time for both $\Rank$ and $\Select$, when $S$ is allowed to change by
insertions or deletions.

Some open problems that remain are:
\begin{enumerate}
\item Is there a succinct indexable dictionary taking ${\cal B}(n,m) + o(n)$
  bits in the RAM model?
\item Is there a representation for $k$-ary trees taking ${\cal
    C}(n,k) + o(n) + O(\lg \lg k)$ bits that can also support subtree
  size operation besides the other navigational operations in constant
  time?
\end{enumerate}

\noindent
{\bf Acknowledgments.} We thank the anonymous TALG referee who helped us
significantly to improve the presentation of the paper.

\end{document}